\documentclass[11pt]{article} 

\usepackage{amsmath,amsthm,amstext,amssymb,amsfonts,amscd,graphicx,bbm,algorithmic,algorithm} 

\newtheorem{lemma}{Lemma}
\newtheorem{thm}{Theorem}

\newtheorem{defn}{Definition}

\newcommand*{\diag}{\mbox{\rm diag}}

\newcommand*{\supp}{\mbox{\rm supp}}

\newcommand*{\0}{\mathbf{0}}
\newcommand*{\1}{\mathbf{1}}

\newcommand*{\minimize}{\text{minimize}}
\newcommand*{\st}{\text{subject to}}
\newcommand*{\R}{\mathbbm{R}}

\title{Existence of Positive Steady States for Mass Conserving and Mass-Action
	Chemical Reaction Networks with a Single Terminal-Linkage Class %
   \thanks{\today.}}

\author{Santiago Akle%
   \thanks{Institute for Computational and Mathematical Engineering,
           Stanford University, Stanford, CA 94305.
					 Research supported in part by the U.S. Department of Energy (Office
					 of Advanced Scientific Computing Research and Office of Biological
					 and Environmental Research) as part of the Scientific Discovery
					 Through Advanced Computing program, grant DE-SC0002009.
           Email: {\tt akle@stanford.edu}, {\tt onkar@stanford.edu}, 
                  {\tt ntaheri@stanford.edu}.}
   \and Onkar Dalal\footnotemark[2]
	 \and Ronan M.T. Fleming\thanks{Center for Systems Biology, University of
	 Iceland, Sturlugata 8, Reykjavik 101, Iceland. Email: {\tt
	 ronan.mt.fleming@gmail.com}}
   \and Michael Saunders\footnotemark[4] 
   \and Nicole Taheri\footnotemark[2]
   \and Yinyu Ye%
   \thanks{Department of Management Science and Engineering,
           Stanford University, Stanford, CA 94305. Research supported in part by NSF grant GOALI 0800151
           and DOE grant DE-SC0002009.
           Email: {\tt saunders@stanford.edu}, {\tt yinyu-ye@stanford.edu}.}}

\date{}

\begin{document}
\maketitle

\begin{abstract}
	We establish that mass conserving, single terminal-linkage networks of
	chemical reactions admit positive steady states regardless of network
	deficiency and the choice of reaction rate constants. This result holds for
	closed systems without material exchange across the boundary, as well as for
	open systems with material exchange at rates that satisfy a simple sufficient
	and necessary condition. 
	
	Our proof uses a fixed point of a novel convex optimization formulation to
	find the steady state behavior of chemical reaction networks that satisfy
	the law of mass-action kinetics. A fixed point iteration can be used to
	compute these steady states, and we show that it converges for weakly
	reversible homogeneous systems. We report the results of our algorithm on
	numerical experiments. 
\end{abstract}

\section{Introduction} 

One of the interests of systems biology is to deterministically model chemical
reaction networks. Models of such systems based on mass-action kinetics depend
on the kinetic parameters of the system.  However, measuring these parameters
experimentally is difficult and error-prone.  Thus, we seek properties of
chemical reaction networks that are independent of kinetic parameters.

In this work, we address the issue of existence of positive steady states,
i.e., positive concentrations of species that will stay constant under the
system's dynamics.  In particular, we tackle the case in which a directed graph
of chemical reactions forms a \emph{strongly connected component} and the
reactions conserve mass. We prove the existence of positive equilibria for such
networks when they are closed systems, and extend our methods to open systems
where complexes are exchanged across the boundary at certain rates.  How best
to compute such a steady state remains uncertain; however, for closed systems
we suggest a fixed point algorithm and provide results of our numerical
experiments for several cases. 

\subsection{Background}
\label{background}

\textit{Chemical Reaction Network Theory} (CRNT), a mathematical theory for
this problem, has its roots in the seminal work by Fritz Horn, Roy Jackson, and
Martin Feinberg in \cite{GMAK, uniqueEPandLyapunov, necc-suff-CB, deficiency0,
deficiency1}. We build on the notation used in these works and summarized by
Gunawardena et al.\ in \cite{gunawardena}. The system consists of a collection
of species reacting collectively in some combination to give another
combination of species in a network of chemical reactions.  Let $\mathcal{S}$
be a set of $m$ species and $\mathcal{C}$ be a set of $n$ complexes. The
relation between species and complexes can be written as a non-negative matrix
$Y\in\R^{m\times n}$, where column $y_j$ represents complex $j$, and $Y_{i j}$
is the multiplicity of species $i$ in complex $j$.  For example, the
multiplicity of the species NaCl in the complex (H$_2$O + 2NaCl) is 2.  

A reaction network is represented by the underlying weighted directed graph
$G(V,E)$, where each node in $V$ represents a complex, each directed edge
$i\rightarrow j$ denotes a reaction using $i$ to generate $j$, and the
positive edge weight $k_{i\rightarrow j}$ is the reaction rate.  The matrix $A
\in \R^{n \times n}$ is the weighted adjacency matrix of the graph, where
$A_{ij}=k_{i\rightarrow j}$.  Define $D := \diag(A\1)$, where $\1$ is the
vector of all ones, and $A_k := A^T-D$.  The elements of this new matrix will
be $(A_k)_{ij} = k_{j \rightarrow i}$ for $i \neq j$, and $(A_k)_{ii} = -\sum_j
k_{j \rightarrow i}$, so that $A_k^T \1 = \0$.  

Let $c \in \R^m$ be the vector of concentrations of each species, and $b \in
\R^m$ the vector of species exchange rates across the network boundary. We
establish necessary and sufficient conditions on the external exchange rates
$b$ in order to determine the existence of steady state concentrations, and to
be able to compute it given a specific $b$.  

We define $\psi(c):\R^m_+\rightarrow \R^n$ to be a nonlinear function that
captures mass-action kinetics:
\[
\psi_j(c) = \prod_i\,c_i^{Y_{ij}}.
\] 
The change in concentration over time can be described by the system of
ordinary differential equations 
\[
\dot{c} = YA_k\psi(c) - b.
\] 
\noindent Hence, steady state concentrations for a chemical reaction network
are any non-negative vector $c^\star\in\R^m$ such that $YA_k\psi(c^\star)=b$.
Equivalently, a pair $(c^\star,v^\star)$ with $c^\star\in\R^m$ and
$v^\star\in\R^n$ will be a steady state if it satisfies the conditions
\begin{align} 
  YA_kv^\star &=b \label{fb}\tag{FB} \\ \psi(c^\star) &= v^\star
  \label{mak}\tag{MA} 
\end{align}
or \emph{flux-balance} and \emph{mass-action}, respectively. Thus, finding 
steady state concentrations is equivalent to finding a vector $v^\star$ that
satisfies both $\eqref{fb}$ and \eqref{mak} for some vector $c^\star\in\R^m$.

Observe that if $(c^\star,v^\star)$ is a positive steady state, then from the
definition of $\psi(c^\star)$,
\begin{align}
  Y^T\log(c^\star)= \log(\psi(c^\star))= \log(v^\star).
  \label{mak-alt}
	\tag{MA-log}
\end{align} 
We refer to this alternative condition as the logarithmic form of the
mass-action condition \eqref{mak}.  In this notation, systems with the property
of mass conservation can be characterized by the following definition. 
\begin{defn}
	A chemical reaction network is $\emph{mass conserving}$ if and only if there
	exists a positive vector $e\in\R^m$ such that 
	\begin{align}
	 e^TYA_k=0,
	  \label{consis}
	\end{align}
	where $e$ denotes the molecular weights of the species (or atomic weights if
	the species are elements).  
\end{defn}

The connectedness of the networks is captured in the following definition of a
terminal-linkage class.  
\begin{defn} 
	A \emph{terminal-linkage class} is defined as a set of complexes
	$\mathcal{L}$ such that for any pair of complexes $(i,j) \in \mathcal{L}$
	there exists a directed path in the graph $G$ that leads from $i$ to $j$.  
\end{defn} 

We further restrict our analysis to a class of weakly reversible networks.
\begin{defn} 
	A chemical reaction network is \emph{weakly reversible} if it is formed
	exclusively by one or more terminal-linkage classes.
\end{defn} 

A reaction network that consists of exactly one terminal-linkage class is
called a \emph{single terminal-linkage network}. Reversibility, at least in a
weak sense, is a prerequisite for steady states with positive concentrations
for all species, as suggested by simple examples like a single non-reversible
reaction.

Next we define a stoichiometric subspace and deficiency for a network.
\begin{defn} 
	A \emph{stoichiometric subspace} $S$ is the subspace defined by the span of
	vectors $y_{j}-y_{i}$, where $y_{j}, y_{i}$ are the columns of $Y$
	representing complexes $i$ and $j$, for each reaction pair $i \rightarrow j$
	in the network. 
\end{defn} 

\begin{defn}
	The \emph{deficiency} of a network is defined as $\delta = n - t - s$, where $t$ is
	the number of terminal-linkage classes and $s$ is the dimension of the
	stoichiometric subspace, also known as the stoichiometric compatibility class.
\end{defn}

In one of their early works, Horn and Jackson \cite{GMAK} analyzed mass-action
kinetics for closed systems (with $b=0$) and defined a class of equilibrium
points called \emph{complex-balanced equilibria}, and defined systems admitting
such an equilibrium to be \emph{complex-balanced systems}. These closed systems
are shown to satisfy the \emph{quasi-thermostatic} and
\emph{quasi-thermodynamic} conditions regardless of the kinetic rate constants.
Following this, Horn \cite{necc-suff-CB} also proved necessary and sufficient
conditions for existence of a \emph{complex-balanced} equilibrium. In
\cite{uniqueEPandLyapunov}, Feinberg and Horn used the existence of a Lyapunov
function to show the uniqueness of the positive steady state in each
stoichiometric compatibility class, which is equivalent to specifying all the
conserved quantities of a system. Later Feinberg \cite{deficiency0,
deficiency1} proved two theorems, now famously known as Deficiency 0-1
theorems, that provide the analysis of positive steady states for a
class of networks with deficiency $0$ or networks with deficiency $1$ but with
each terminal-linkage class having deficiency less than $1$. For this
restricted class of closed systems, the existence of a positive steady state is
given by Perron-Frobenius theory for a positive eigenvector.  Other work in
this area is from the perspective of dynamical systems and aimed toward proving
two open conjectures: \emph{Global Attractor Conjecture} and \emph{Persistence
Conjecture} \cite{Anderson-GAC}. Another approach using parametrized convex
optimization to compute a non-equilibrium steady state is given in
\cite{fleming-opt}.

To the best of our knowledge, the vast majority of CRNT research studies closed
\textit{complex-balanced systems}, which, by definition, admit a
\textit{complex-balanced equilibrium}. In the notation above, a
complex-balanced equilibrium exists when the vector of concentrations $c$ satisfies
$A_k\psi(c)=0$, i.e., the vector $\psi(c)$ belongs to the null space of $A_k$.
It can be shown that a network with linearly independent complexes will have
deficiency $\delta = 0$. However, if some of the complexes are linearly
dependent (as shown in the example in Section 3), there are systems that are
not complex-balanced yet admit concentrations in equilibria where $A_k \psi(c)$
is in the null space of $Y$ and $A_k\psi(c) \neq 0$.  Though the condition of
\emph{complex-balance} is sufficient for \emph{thermodynamic} consistency,
\cite{GMAK} shows that it is not necessary. Also, for open systems with
material exchange across the boundary, \emph{complex-balance} is not defined.
In order to handle open systems, these works hint at extending the system using
a \emph{pseudo $0$-complex} and adding \emph{pseudo reactions}.  However, it is
unclear how to choose the \emph{pseudo kinetic rates} such that the positive
eigenvector solution of the extended system will achieve the given external
exchange rates $b$. From the point of view of systems biology and bio-chemical
engineering, analyzing the behavior of a cell under different exchange
conditions $b$ is very important to control and engineer the cell, for example
studying the desired effects in pharmacology, or producing specific metabolites
in bioreactors.

\vspace{0.2in}
In this paper, we extend the previous work on two accounts: 1) we prove
existence of positive equilibria in closed systems for some reaction networks
that do not satisfy the necessary conditions of the Deficiency 0-1 theorems
(are not necessarily \emph{complex-balanced}), and 2) we provide a necessary
and sufficient condition on the external exchange rate $b$ for some open
systems to admit a positive steady state. We use a fixed point of a convex
optimization problem, with an objective function similar to the Helmholtz
function defined in \cite{GMAK}. The fixed point of this mapping gives the
required steady state. We prove the existence of a positive steady state for
any weakly reversible chemical reaction network with a \emph{single
terminal-linkage class}. We strongly believe that this can be extended to
systems with \emph{multiple terminal-linkage classes}, as supported by our
computational results for randomly generated networks. Section
\ref{example-network} gives a detailed analysis of a toy network to emphasize
this claim.

\section{A Fixed Point Model} 
\label{section:fp-model}

Our main result establishes that for any set of positive reaction rates
$k\in\R^n$ and any $b$ in the range of $YA_k$, a single terminal-linkage
network will admit a positive solution pair $(c,v)$ that satisfies the laws
\eqref{fb} and \eqref{mak}.  We show this by defining a positive fixed point of
a convex optimization problem, and establishing an equivalence between the
positive fixed point and positive solution to the equations.

We construct a fixed point mapping of a linearly constrained optimization
problem such that the logarithmic form of the mass-action equation
\eqref{mak-alt} is an optimality condition, and hence any solution to this
optimization problem will also satisfy \eqref{mak-alt}.

To define the mapping, let $b=YA_k\eta$ for some $\eta \in \R^n$ and observe
that for arbitrary $s \in \R^n$ we can write $b = YA^T(\eta +s) - YD(\eta +
D^{-1}A^Ts)$.  In particular, we choose $s$ positive and large enough so
that $\eta^+ := \eta+s$ and $\eta^- := \eta + D^{-1}A^Ts$ are both positive.
Also, from Definition \ref{consis}, $e^Tb = 0$ and thus
\begin{equation}
  e^TYD\eta^- = e^TYA^T\eta^+.
  \label{massbalanceb}
\end{equation}
Define $\mu:=(r,r_0) \in \R^{m+1}$ to be a vector parameter.  Observe that if
the parametric convex optimization problem
\begin{equation}
	\begin{array}{lll}
  \underset{(v,v_0)\in\R^{n+1}}{\minimize} & v^T D(\log(v)-\1) + v_0(\log v_0 -1) \\
	\st &  YD v + YA^T\eta^+ v_0 = YA^Tr + YD\eta^-r_0 &:\ y \\
	    &\quad (v,v_0) \ge 0                     
	\end{array}
	\label{convex-fix}
\end{equation}
has a positive solution $(v^\star(\mu),v_0^\star(\mu))$, then the optimality
conditions
\begin{equation}
	\begin{array}{rl}
	YDv^\star(\mu)  + YA^T\eta^+ v^\star_0  &= YA^Tr + YD\eta^-r_0  \\
	DY^T y^\star(\mu) &= D\log(v^\star(\mu))  \\
	(YA^T\eta^+)^Ty^\star(\mu)   &= \log(v^\star_0(\mu)) \\
	(v^\star(\mu),v_0^\star(\mu)) &\ge 0
	\end{array}
	\label{optcon}
\end{equation}
are well defined. Since $D$ is nonsingular, the second optimality condition is
equivalent to \eqref{mak-alt}, for $c^\star(\mu) := e^{y^\star(\mu)}$, where
the exponent is taken element-wise.  Hence, the equation \eqref{mak-alt} holds
and $c^\star(\mu)$ satisfies mass-action. We show that for some parameter
$\mu$, the equality 
\begin{equation}
 YA_kv^\star=v^\star_0b
  \label{scaled-fb}
\end{equation} is
satisfied. This implies that for $b=0$, both \eqref{fb} and \eqref{mak-alt} are
satisfied, and the solution is attained. For the case where $b\neq 0$, we can
construct a corresponding solution so that $\eqref{fb}$ holds.

Note that the nonlinear program \eqref{convex-fix} is strictly convex, so for
any feasible $(r,r_0)$ there is a unique minimizer. That is, the mapping 
\begin{equation}
  (r,r_0) \rightarrow (v^\star(\mu),v^\star_0(\mu))
  \label{mapping}
\end{equation} 
is well defined.  If $\hat \mu = (r, r_0)$ is a fixed point of \eqref{mapping}, then
the linear equality constraint in \eqref{convex-fix} implies
\[
	 YDv^\star(\hat \mu)+YA^T\eta^+v^\star_0(\hat \mu) = YA^Tv^\star(\hat \mu) +
	 	YD\eta^-v_0^\star(\hat \mu)
\]
or, equivalently,
\[
    Y A_k v^\star(\hat \mu) = Y(A^T-D)v^\star(\hat \mu) = v^\star_0(YA^T\eta^+ - 
			YD\eta^-) = v^\star_0(\hat \mu)b.
\]
Therefore, if such a fixed point exists, the solution $v^\star(\hat\mu)$ at this
fixed point will satisfy $\eqref{scaled-fb}$.  For simplicity, we henceforth
refer to the optimal solution variables $(v^\star(\mu),v^\star_0(\mu)),
y^\star(\mu)$ as $(v^\star,v^\star_0), y^\star$, but acknowledge their
dependence on $\mu$.

\begin{thm} 
	\label{fp-exist-map} 
	For any mass conserving, mass-action chemical reaction network and any choice
	of rate constants $k>0$, there exist nontrivial fixed points for the mapping
	\eqref{mapping}.
\end{thm} 

\begin{proof} 
	Brouwer's fixed point theorem states that any continuous mapping from a
	convex and compact subset of a Euclidean space $\Omega$ to itself must have
	at least one fixed point. 

	Let $(v^\star,v^\star_0)$ be defined as in \eqref{mapping} and let $\gamma$
	be a positive fixed scalar. Define the set 
	\[
	\Omega = \left\{ (v,v_0)\in	\Re^{n+1} \; : \; (v,v_0)\geq 0, 
		\quad e^TYDv + e^TYA^T\eta^+v_0 = \gamma\right\},
	\]  
	where $e$ is defined in \eqref{consis}.  According to Brouwer's fixed point
	theorem, if the parameter $(r,r_0) \in \Omega$ ensures that the corresponding
	solution to the optimization problem $(v^\star,v_0^\star) \in \Omega$, then
	there is a fixed point such that the parameter and the solution are equal,
	i.e., there exists a $\mu$ such that $\mu = (r,r_0) =
	(v^\star(\mu),v_0^\star(\mu))$.  

	The set $\Omega$ is bounded and formed by an intersection of closed convex
	sets, and hence is convex and compact.  Moreover, the mapping $\mu\rightarrow
	(v^\star,v^\star_0)$ is continuous. Since problem \eqref{convex-fix} is
	feasible for any $\mu\in\Omega$, the mapping $\Omega \ni \mu \rightarrow
	(v^\star,v^\star_0)$ is well defined.

	To show that the image of $\Omega$ under the mapping $(r,r_0) \rightarrow
	(v^\star,v_0^\star)$	is in $\Omega$, first observe that by the bounds in
	\eqref{convex-fix}, $(v^\star,v^\star_0) \ge 0$.  Using the equality
	constraints, Definition \eqref{consis} and Equation \eqref{massbalanceb},
	we have
	\begin{align}
	e^T YD v^\star + e^T YA^T \eta^+ v^\star_0 &= e^T YA^T r + e^T YD \eta^- r_0 \notag 
	\\ & = e^T YD r + e^T YA^T \eta^+r_0=\gamma, \notag 
	\end{align} 
	and thus $(v^\star,v_0^\star)\in \Omega.$

	Therefore, under the mapping $(r,r_0) \rightarrow (v^\star,v_0^\star)$,
	$(r,r_0) \in \Omega$ implies $(v^\star,v^\star_0) \in \Omega$, and the
	mapping must have a fixed point.  Moreover, since $\Omega$ does not contain
	the zero vector, the fixed point(s) are nontrivial.
 
\end{proof}

Note that the value of $YDv^\star+YA^T\eta^+v^\star_0$ is the rate of
consumption of each chemical species and $YA^Tv^\star+YD\eta^-v^\star_0$ is the
rate of production of each chemical species. At the fixed point, the equality
$YDv^\star + YA^T\eta^+v^\star_0= YA^Tv^\star+YD\eta^-v^\star_0$ defines a
steady state. The set $\Omega$ defines the parameter $\gamma=e^T(YDv^\star +
YA^T\eta^+v^\star_0)$; since the vector $e$ can be interpreted as an assignment
of relative mass to the species, $\gamma$ can be interpreted as the total
amount of mass that reacts per unit time at the steady state.  Therefore,
looking for fixed points in $\Omega$ corresponds to looking for steady states
where the amount of mass that reacts in the system is prescribed.

We have established the existence of a nontrivial fixed point $\mu$ of the
mapping $\Omega \ni \mu \rightarrow (v^\star,v^\star_0)\in \Omega$. Moreover,
we have shown that when the associated minimizer $(v^\star,v^\star_0)$ is
positive, it is a solution to \eqref{mak} and to
$YA_kv^\star=v^\star_0b$.  However, in the case when some entries of $v^\star$
are zero, the objective function of \eqref{convex-fix} is non-differentiable
and we cannot use the optimality conditions to show that \eqref{mak} holds.

\subsection{Positive fixed points in single terminal-linkage networks}  
	\label{section::single-linkage}

We now consider the case when the network is formed by a single
terminal-linkage class and show that if $\hat \mu$ is a fixed point of the
mapping \eqref{mapping}, the minimizer $(v^\star(\hat
\mu),v^\star_0(\hat\mu))$, and therefore $\hat\mu$, is positive.

Lemma $\ref{maximum-support}$ shows that if problem $\eqref{convex-fix}$ has a
feasible point with support $J$, the minimizer $(v^\star,v^\star_0)$ will have
support at least $J$. Lemma \ref{positive-feasible} uses the single
terminal-linkage class hypothesis to show that at a fixed point, there is a
positive feasible point.  These two Lemmas imply that at a fixed point
$\hat{\mu}$, the minimizer will be positive.  Finally Theorem \ref{thm:scaling}
shows that if $\hat{v}_0 \neq 1$ at the solution, we can construct another
solution for which $\hat{v}_0=1$. This establishes that there is a nontrivial
steady state for the network. 

To complete the argument we must prove Lemmas \ref{maximum-support},
\ref{positive-feasible} and Theorem \ref{thm:scaling}.

\begin{lemma} 
	The support of any feasible point of Problem \eqref{convex-fix} is a subset
	of the support of the minimizer $(v^\star,v^\star_0)$. 
	\label{maximum-support}
\end{lemma}
\begin{proof}
	Let $\tilde{v}\in \Re^{n+1}$ be any of the feasible points with the largest
	support and let $z$ be any feasible direction at $\tilde{v}$. By
	construction, for all $\alpha$ in some interval $[\ell,u]$ the points
	$v_\alpha:= \tilde{v} + \alpha z$ are non-negative and feasible.  The
	interval can be chosen so that when $\alpha=l$ and when $\alpha=u$, one new
	bound constraint becomes active.  This implies that $\supp(v_\ell)$ and
	$\supp(v_u)$ are strictly contained in $\supp(\tilde{v})$, and
	$\supp(v_\alpha) = \supp(\tilde{v})$ for $\alpha \in (\ell,u)$.
	
	Without loss of generality, we assume $\ell<0<u$, since $\ell$ and $u$ will
	not be of the same sign; if $\ell = 0$ and $u>0$, any point $v_{\alpha}$ can
	be written as a convex combination of $\tilde{v}$ and $\tilde{v}+uz$, and
	thus has support as large as $\tilde{v}$.

	Define the univariate function 
	\begin{equation}   
		g(\alpha) := \phi(\tilde{v} + \alpha z), 
		\label{univariate}
	\end{equation} 
	where $\phi$ is the objective function of \eqref{convex-fix}.  We will
	establish that as $\alpha\rightarrow l$ the derivative $g'(\alpha)\rightarrow
	-\infty$, and as $\alpha\rightarrow u$ the derivative $g'(\alpha)\rightarrow
	\infty$.  Thus, by the mean value theorem, there must exist a zero of the 
	function $g$ in the	interior of the interval $[l,u]$.  Since this function is
	strictly convex, if a stationary point exists in the interior of the
	interval, the function value at the stationary point must be smaller than at
	the boundary.

	Observe that if we let $d_i$, for $i\in[1,\dots n]$, be the diagonal entries
	of $D$ and $d_{n+1}=1$, we can write 
  \begin{align*}
	 g(\alpha) &=\sum_{i=1}^{n+1} (\tilde{v} + \alpha z_i) d_i\log(\tilde{v}_i + \alpha z_i).
  \end{align*} 

	An important observation is that if some entry $\tilde{v}_j=0$ then $z_j=0$,
	otherwise $v_\alpha$ would have a larger support for some $\alpha \neq 0$.
	This implies that $(v_\alpha)_j=0$ for all entries where $\tilde{v}_j=0$.  If
	we let $J$ be the set of nonzero entries of $\tilde{v}$, and $L$ be the
	subset of $J$ formed by the entries that tend to zero as $\alpha \rightarrow
	l$, then 
   \begin{align*} 
	 	g'(\alpha) &=\sum_{i\in J} z_i d_i(\log(\tilde{v}_i + \alpha z_i)) \\            
		&= \sum_{i\in (L^c\cap J)} z_id_i(\log{(\tilde{v_i}+\alpha z_i)}) + \sum_{i\in L}
		 z_id_i(\log{(\tilde{v_i}+\alpha z_i)}). 
	\end{align*} 
	As $\alpha \rightarrow l$, the first summation will approach a finite value.
	Since $z_i>0$ for all $i\in L$, the entries in the logarithm of the second
	sum tend to zero and the term will diverge to $-\infty$.

	Similarly, let $U$ be the subset of $J$ formed by the entries that tend to zero as 
	$\alpha\rightarrow u$. Observe that for these entries, $z_i<0$ and
  \[
    g'(\alpha) = \sum_{i\in (U^c\cap J)} z_i d_i(\log{(\tilde{v_i}+\alpha z_i)})
               + \sum_{i\in U}   z_i d_i(\log{(\tilde{v_i}+\alpha z_i)}).
  \]
	The first sum will tend to a finite value and the second will diverge to
	$\infty$. 
	
	Now, assume that for some $\mu$ there is a feasible point
	$(\tilde{v},\tilde{v_0})$ with larger support than the minimizer
	$(v^\star,v^\star_0)$ of problem \eqref{convex-fix}.  Since
	$(v^\star,v_0^\star)$ has smaller support, we can write  $(v^\star,
	v_0^\star) = (\tilde{v},\tilde{v}_0) + \alpha^\star z$ where $\alpha^\star$
	is on the boundary of the corresponding feasible interval.  By the
	previous argument, there is a value of $\hat \alpha \neq \alpha^\star$ in the
	interior of the interval such that $(v^\star, v_0^\star) =
	(\tilde{v},\tilde{v}_0) + \hat \alpha z$ has a lower function value than
	$(v^\star,v^\star_0)$, contradicting its optimality.

	Therefore, by the mean value theorm, there must exist a stationary point of
	$g$ strictly in the interior of the interval $[l,u]$ at which the function
	value is smaller than at the boundary.  Moreover, the optimal point will have
	at least the support of any feasible point.  

  \end{proof}

\begin{lemma}
	If the network is formed by a single terminal-linkage class, when Problem
	\eqref{convex-fix} is parametrized by a fixed point $\hat\mu$,  there
	exists a positive feasible point $(\hat v,\hat v_0)$.
\label{positive-feasible}
\end{lemma}

\begin{proof}
	Let Problem \eqref{convex-fix} be parametrized with a fixed point $\hat\mu$,
	and let $(\hat v, \hat v_0)$ be both the minimizer and the fixed point.  We
	prove by contradiction that no entry of the minimizer $(\hat{v},\hat{v}_0)$
	can be zero. Observe that by the definition of $\Omega$ the origin is not
	contained in the set, and therefore the fixed point cannot be identically
	zero. 

	First, assume that $\hat{v}_0>0$ and observe that $\rho := (D^{-1}A^T\hat v + \eta^-\hat v_0, 0)$ is 
	a feasible point. Since $\eta^-$ was chosen to be positive and 
	$D^{-1}A^T$ has no zero columns, the support of $\rho$ are the first $n$ entries 
	of the vector. A convex combination of $(\hat v,\hat v_0)$ and $\rho$ will 
	be feasible and have full support.

	Now, assume that $\hat v_0 = 0$ and some entry of $\hat v$ is nonzero, and
	observe that $(D^{-1}A^T\hat{v},0)$ is feasible. A convex combination of
	$(D^{-1}A^T\hat{v},0)$ and $(\hat{v},0)$ is feasible and its support
	contains the union of the supports of the two vectors.  That is, for $\beta \in
	[0,1]$ the point  $(\tilde{v}, \tilde{v}_0) = \beta (D^{-1}A^T\hat{v},0) +
	(1-\beta)(\hat{v},0)$ is feasible, and using the fact that the support of
	$D^{-1}A^T\hat v$ is the support of $A^T	\hat v$ along with Lemma
	\ref{maximum-support},
	\begin{equation} 
		\left(\supp(\hat{v},0) \cup \supp(A^T\hat{v},0)\right) \subset \supp(\tilde{v},0).
	  \label{monotonesets} 
	\end{equation} 
	
	This relation can be used inductively to show that there is a feasible point
	with support at least as large as the union of the supports of
	$((A^T)^p\hat{v}, 0)$ for all positive powers of $p$.

	The single terminal-linkage class hypothesis implies that for any pair
	$(i,j)\in [1,\dots,n]\times [1,\dots,n]$, there exists a power $p$ large
	enough such that $(A^T)^p_{ij}>0$.  More importantly, if $\hat v_j > 0$, then
	there exists a $p$ such that $[(A^T)^p \hat v_j]_i > 0$ for all $i$.  
	Therefore, if $\hat{v}_0 = 0$ and $\hat{v} \neq 0$, there is a 
	feasible point $(\tilde{v},0)$ such that $\tilde{v} > 0$.  
	
	Finally, if $\hat v > 0$ then there is a scalar $0<\alpha$ small enough such
	that 
	\[
	0< \hat v - D^{-1}A^T\eta^+ \alpha,
	\]
	and then the equality  
	\[
		YD(\hat v - D^{-1}A^T\eta^+\alpha)+YA^T\eta^+\alpha = YD\hat v = YA^T\hat v
	\]
	implies that the positive point $(\hat v - D^{-1}A^T\eta^+\alpha,\alpha)$ is
	feasible.

	Therefore, if a network is formed by a single terminal-linkage class, the
	problem \eqref{convex-fix} has a positive feasible point $(\hat v, \hat
	v_0)$.  
\end{proof}

\begin{thm}\label{thm:scaling}
	For a mass conserving single terminal-linkage network, there exists a
	concentration $c>0$ such that $YA_k\psi(c) = b$ if and only if $b$ is in the
	range of $YA_k$. 
\end{thm}
\begin{proof}
	We have shown that there exist positive vectors $c \in \R^m, v \in \R^n$ such
	that $Y^T \log(c) = \log (v)$ and $Y A_k v = v_0 b$.  In other words, we have
	proven that there is a positive vector $c$ and a positive scalar $\alpha$
	such that $Y A_k \psi(c) = \alpha b$.  If we can construct a new
	concentration vector $\tilde c > 0$ that satisfies $\psi(\tilde c) =
	\frac{1}{\alpha}\psi(c)$, then
	\[
		YA_k\psi(\tilde c) = \frac{1}{\alpha}YA_k\psi(c) = b
	\]
	and the steady state concentration $\tilde c$ satisfies \eqref{fb} and
	\eqref{mak}.  

	First, we argue that the vector of all ones, $\1 \in \R^n$, is in the range
	of $Y^T$ when the network consists of a single terminal-linkage class and
	is mass conserving. The condition of mass conservation implies that $e^T Y A_k = 0$, 
	or equivalently, $Y^T e \in \mathcal{N} (A_k^T)$.  Since $A_k^T$ is the 
	Laplacian matrix of a strongly connected graph, $\mathcal{N}(A_k) = \{\beta
	\1 :\; \beta \in \R\}$; thus, for some value $\hat \beta$, $Y^T e = \hat\beta
	\1$. 
	
	Now, observe that $\log\left(\frac{1}{\alpha}\psi(c) \right) =
	\log(\psi(c))-\1\log(\alpha) = Y^T\log(c) - \1\log(\alpha)$, where $\log(c)$
	is an entry-wise logarithm of the vector and $\log(\alpha)$ the scalar
	logarithm. Moreover, since $\1$ is in the range of $Y^T$, say $Y^T \delta =
	\1$ for some $\delta \in \R^m$, then $\log(\alpha)\1 =
	Y^T(\log(\alpha)\delta)$. Thus, if we define $\tilde c$ as the vector that
	satisfies  $\log(\tilde c) = \log(c) - \log(\alpha)\delta$ , then 
	\[
	Y^T\log(\tilde c) = Y^T\log(c) - \log(\alpha)\1,
	\] 
	which implies that 
	\[
	\psi(\tilde c) = \frac{1}{\alpha}\psi(c).
	\] 
	Therefore, the inhomogeneous system has a solution if the underlying graph of
	the network is formed by a single terminal-linkage class and the network is 
	mass conserving, regardless of the kinetic parameters.

	\end{proof}

\section{Numerical Experiments}

The results in this section depend on the calculation of positive steady states
for weakly reversible networks. We use the algorithm described below to solve
for the associated fixed points.

\subsection{Numerical method for finding fixed points}

Given an initial positive point $\hat v \in \R^{n}$ and a small tolerance
$\tau$, we use the following fixed point iteration, Algorithm
\ref{fp-iteration}, to find a parameter $\mu = (r, r_0)$ to the problem
\eqref{convex-fix} such that $(v^\star,v_0^\star) = (r,r_0)$.

\begin{algorithm}
\caption{Fixed Point Iteration to Find a Steady-State Concentration}
\label{fp-iteration}
\begin{algorithmic}[1]
  \STATE $(v^\star,v_0^\star) \gets (\hat v,0)$
  \STATE $(r,r_0) \gets (v^\star,v_0^\star)$
  \WHILE{$\|YA_k v^\star - (Y A_k \eta^+ + YD \eta^-) v_0^\star \|_\infty>\tau$}
  \STATE $(v^\star,v^\star_0)\gets \text{unique solution of \eqref{convex-fix}} $
		\label{min-step}
	\STATE $(r,r_0) \gets \frac{1}{2}(r,r_0) +\frac{1}{2}(v^\star,v^\star_0)$
  \ENDWHILE
  \label{fixpoint-alg}.
\end{algorithmic}
\end{algorithm}

Step \ref{min-step} in the while loop of Algorithm \ref{fp-iteration} requires
solving the linearly constrained convex optimization problem
\eqref{convex-fix}. Our implementation uses the PDCO package \cite{pdco} to
solve this problem.

Provided that at each iteration $k$, the unique solution of \eqref{convex-fix}
satisfies $v^\star(\mu_k)>0$ and the minimization is solved with sufficient
accuracy, the optimality conditions for \eqref{convex-fix} will imply that for
all $k = 1, 2, 3, \ldots$,
\[
	\|Y^T\lambda^\star-\log(v^\star(\mu_k))\|_\infty \leq \epsilon,
\] 
for some small value of $\epsilon$, where $\lambda^\star$ is the Lagrange
multiplier of the linear equality constraint at the solution that corresponds
to the logarithm of the concentrations. Thus, if the iteration
converges to a fixed point $(r,r_0) = (v^\star,v_0^\star)$, then at this fixed
point \eqref{mak} will be satisfied to a precision $\epsilon$ and \eqref{fb}
will be satisfied to a precision $\tau$.

Algorithm \ref{fp-iteration} has been tested extensively on randomly generated
networks with noteworthy success. The results of our experiments are shown in
Section \ref{scn:convergence}.

\subsection{An example network}
\label{example-network}

\begin{figure}[!htbp]
\centering
	\includegraphics[width=3in]{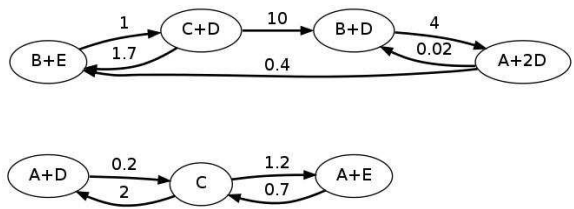}
	\caption{Example network with two terminal-linkage classes}
	\label{fig:network-small}
\end{figure}

In this section we consider the toy network shown in Figure
\ref{fig:network-small}.  The number of complexes is $n = 7$, the number of
terminal-linkage classes is $l = 2$, and the stoichiometric subspace $S =
\operatorname{span}\{ A+E-C, C-A-D, B-C\}$ has dimension $s = 3$.  Therefore,
the deficiency of this network is $\delta = 7 - 2 - 3 = 2$, and hence neither
of the Deficiency 0-1 theorems \cite{deficiency0, deficiency1} can be applied
to calculate equilibrium points. 

However, since this network is weakly reversible, intuition suggest that a
non-zero steady state exists.  We use Algorithm \ref{fp-iteration} to solve for
the fixed point described in Section \ref{section:fp-model}, obtaining a
positive steady state.  Figure \ref{fig:ConcentrationVsIteration} illustrates
the convergence of the fixed point iterations to the steady state.\\

\begin{figure}[!htbp]
	\includegraphics[width=4.5in]{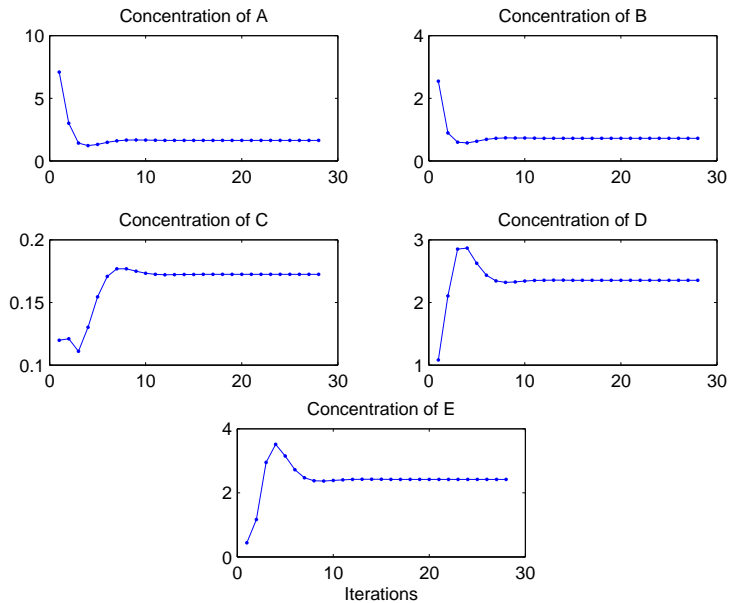}
\caption{Concentration convergence with iteration}
\label{fig:ConcentrationVsIteration}
\end{figure}

Figure \ref{EquilibriumVsTotalMass} illustrates the change in steady state as a
function of the \emph{total mass} in the system, where total mass is defined as
$\gamma = e^T(YD v^\star + YA^T \eta^+ v^\star_0)$, as described in Theorem
\ref{fp-exist-map}.  The experiment shows that as the total mass increases,
species $A$, $B$ and $C$ adjust linearly to the additional mass, while species
$D$ and $E$ stay at the same levels.  This linear growth in species $A, B$ and
$C$ can be explained analytically by the fact that the vector $\mathbf{1}$ lies
in range of $Y^{T}$. 

\begin{figure}[!htbp]
  \includegraphics[width=4.5in]{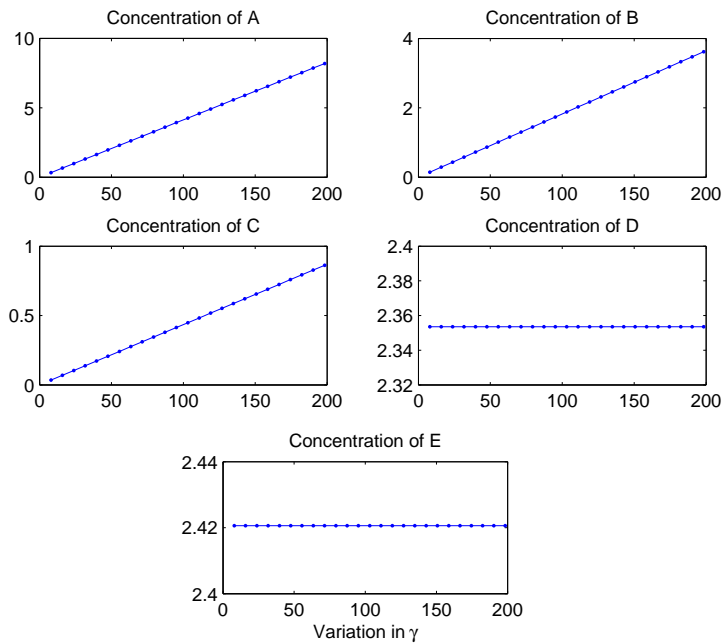}
	\caption{Equilibrium dependence on total mass} 
	\label{EquilibriumVsTotalMass}
\end{figure}

\subsection{Generating suitable networks}

This section describes the sampling scheme used to generate random mass
conserving chemical reaction networks with a prescribed number of
terminal-linkage classes.  The output of the method is a network with $n$
complexes, $m$ species, and $\ell$ strongly connected components, where $\ell$
is the desired number of terminal-linkage classes.

First, we iteratively generate Erd\H{o}s-R\'{e}yni graphs%
\footnote{An Erd\H{o}s-R\'{e}yni graph is a directed unweighted graph. Each
edge is included with probability $p$ and all edges are sampled \emph{iid}.} %
with $m$ nodes until we sample a graph with $\ell$ strongly connected
components; call this graph $\hat G( \hat V, \hat E)$.  We give each edge in 
$\hat E$ a weight of an independent and uniformly distributed value in the
range $(0,10]$.  These edge weights represent the reaction rates between
complexes.

To generate the stoichiometry, we define a parameter $r$ as the maximum 
number of species in each complex. Each complex is constructed with a random
sample of $q$ species, where $q$ is a random integer in $[1,r]$.  All samples
are done uniformly and independently.  Finally, we assign the multiplicity of
each species in a complex with independent samples of the absolute value of a
standard normal unit variance distribution. To ensure mass is conserved, we
normalize the sum of the stoichiometry of the species that participate in a
complex to one, so that $Y^T\1 = \1$ and $A_k^TY^T \1 = \0$. 

\subsection{Convergence of the Fixed Point Algorithm}
\label{scn:convergence} 

This section illustrates the convergence of Algorithm \ref{fp-iteration} on
large networks that consist of either a single terminal-linkage class 
or multiple terminal-linkage classes, i.e., weakly reversible networks.

Algorithm \ref{fp-iteration} produces sequences that, up to a small tolerance,
satisfy \eqref{mak} at every iteration. Ideally, the infeasibility with 
respect to \eqref{fb} also monotonically decreases until convergence.  Our
extensive numerical experiments indicate that this is in fact the behavior for
homogeneous networks that are weakly reversible.

\begin{figure}[!htbp] 
	\centering 
	\includegraphics{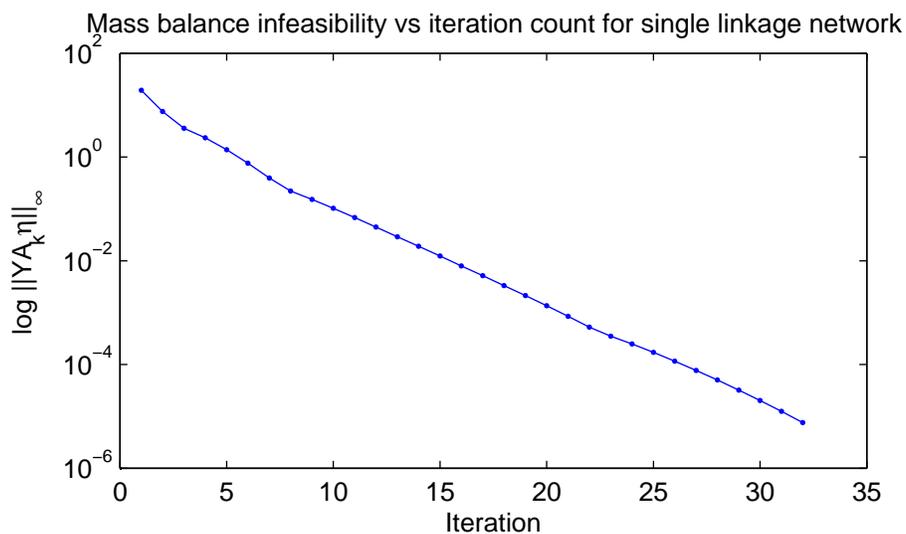}
  \caption{Typical Infeasibility of \eqref{fb} vs.\ Iteration, for network with a
  single terminal-linkage class} 
	\label{fig:typical-infeas-single} 
\end{figure}

\begin{figure}[!htbp]
  \centering
	\includegraphics{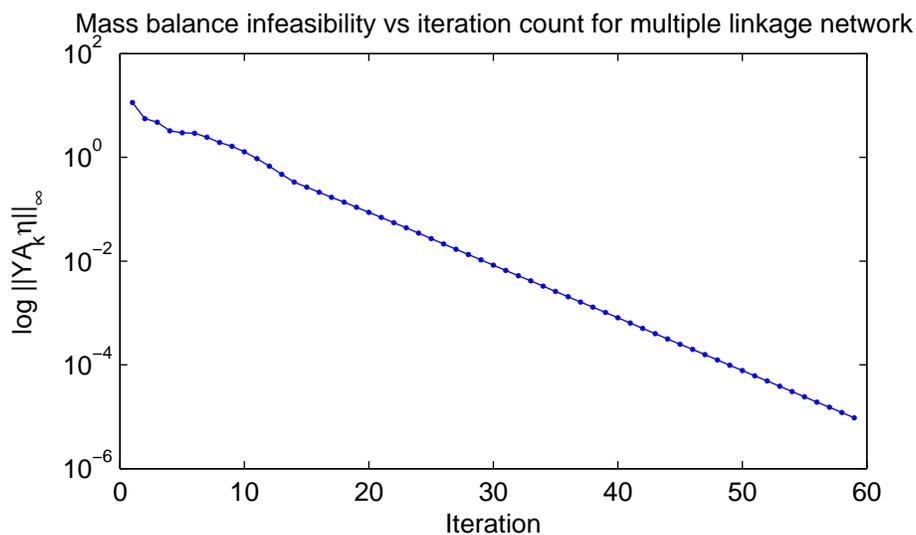}
	\caption{Typical Reduction in Infeasibility of \eqref{fb}, for network with
	two terminal-linkage classes}
	\label{fig:typical-infeas-multiple}
\end{figure}

Figure \ref{fig:typical-infeas-single} displays the sequence of the
infeasibilities $\|YA_kv_k\|_\infty$ at each iteration $k$ in Algorithm 
\ref{fp-iteration}, for a network with a single terminal-linkage class, $50$
species and $500$ complexes, where at most $10$ species participate in each
complex. Figure \ref{fig:typical-infeas-multiple} displays the analogous
sequence for a network of equal size and two terminal-linkage classes.  We have
observed this (apparently linear) convergence rate consistently over all
generated networks, regardless of the number of terminal-linkage classes.

\begin{figure}[!htbp] 
	\centering
  \includegraphics{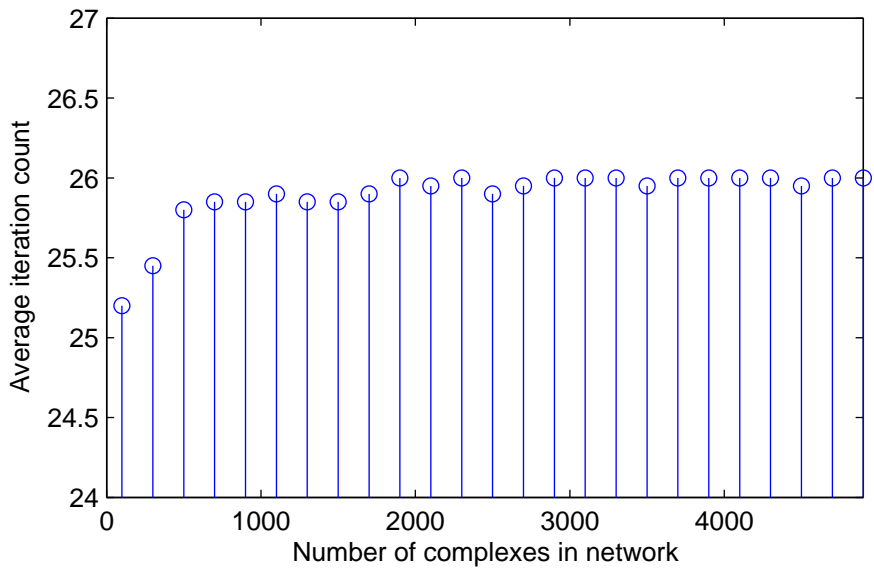} 
	\caption{Average number of iterations for single terminal-linkage class
		networks}
  \label{fig:iteration-count-simple} 
\end{figure}

\begin{figure}[!htbp] 
	\centering
  \includegraphics{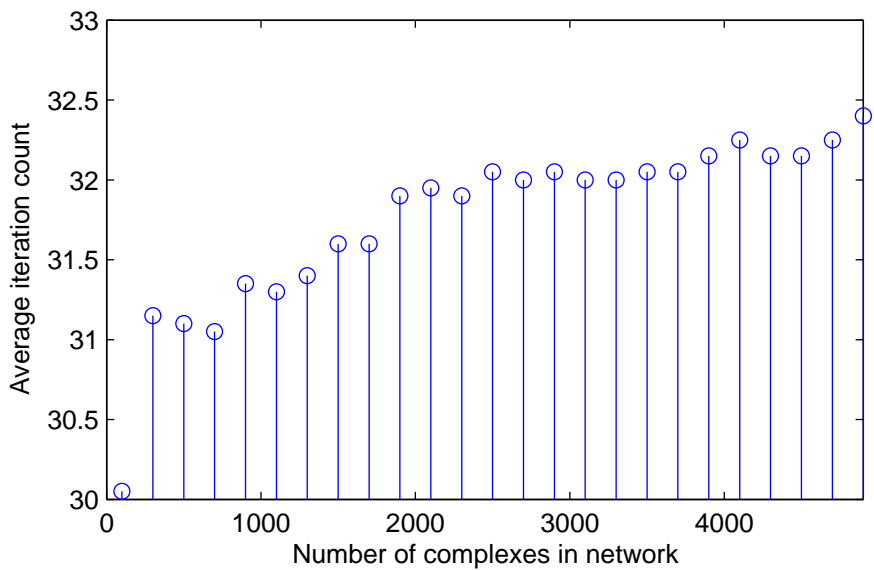} 
	\caption{Average number of iterations for networks with two terminal-linkage
		classes}
  \label{fig:iteration-count-multiple} 
\end{figure}

We have also investigated the number of iterations necessary for Algorithm
\ref{fp-iteration} to converge on networks of different sizes, with either one
or two terminal-linkage classes.  Figures \ref{fig:iteration-count-simple} and
\ref{fig:iteration-count-multiple} display the mean number of iterations
necessary for convergence on networks ranging from $100$ to $5000$ complexes,
where each average is taken over $20$ instances per network size. Notably, the
average number of iterations increases less than $10\%$ as the network size
grows fifty-fold.

In future work, we plan to prove theoretical results on the existence of
positive equilibria for chemical reaction networks with multiple
terminal-linkage classes.  However, our comprehensive numerical experiments
seem to indicate that even for networks with more than one terminal-linkage
class, there exists at least one positive fixed point of problem
\eqref{convex-fix}, and the iterates of Algorithm \ref{fp-iteration} converge
to such a fixed point.

Finally, while conducting this work, we were made aware of related work by Deng
et al.\ \cite{Deng}. Their work extends that of Feinberg and Horn by proving
that weak reversibility is a necessary and sufficient condition for existence
of positive equilibria. While their proof is for closed systems (with $b=0$),
it is not clear how to use the $0$-\emph{complex} in their extension to
solve for any admissible $b$, as mentioned in Section \ref{background}.  More
importantly, our proof is based on a less complicated convex optimization
formulation, which gives a method to calculate numerical solutions using a
fixed point algorithm. \\

\hspace{-0.2in}%
\textbf{Acknowledgment}: We gratefully acknowledge Anne Shiu and her student
for helpful discussions and directing us to \cite{Deng}.

\frenchspacing
\nocite{*} 
\bibliographystyle{plain} 
\bibliography{fixpoint}{}

\end{document}